\documentclass[11pt]{article}

\usepackage[a4paper,margin=30mm]{geometry}
\usepackage{amsmath,amssymb,amsthm,mathtools}
\usepackage{booktabs}
\usepackage[numbers,sort&compress]{natbib}
\usepackage[hidelinks]{hyperref}

\newcommand{\RR}{\mathbb R}
\newcommand{\Sym}{\operatorname{Sym}}
\newcommand{\Sp}{\operatorname{Sp}}
\newcommand{\Tr}{\operatorname{Tr}}
\newcommand{\Id}{\mathbb I}
\newcommand{\Om}{\Omega}

\newtheorem{theorem}{Theorem}[section]
\newtheorem{proposition}[theorem]{Proposition}
\newtheorem{lemma}[theorem]{Lemma}
\newtheorem{corollary}[theorem]{Corollary}
\theoremstyle{remark}
\newtheorem{remark}[theorem]{Remark}

\title{The Killing Form and Petz Uniqueness of Gaussian Bures Geometry}
\author{Christian Kerskens\\
\small Trinity College Institute of Neuroscience, Trinity College Dublin,
Dublin, Ireland\\
}
\date{}

\begin{document}
\maketitle

\begin{abstract}
For centered bosonic Gaussian states, the covariance pullback of the Bures
cometric differs from the lift-normalized classical covariance
Fisher--Rao cometric by a state-independent bilinear form.  Under the
canonical identification of symmetric covariance covectors with
\(\mathfrak{sp}(2N,\RR)\), this form is the trace form and hence a fixed
multiple of the Killing form.  We prove that, within the normalized
symmetric Petz family, Bures is uniquely selected by requiring such an
additive state-independent symplectic correction.  We determine the full
Williamson-frame spectrum and show that a boundary stratum with \(m\) pure
modes has an \(m^2\)-dimensional cometric kernel isomorphic to
\(\mathfrak u(m)\), while the pure Gaussian orbit remains nondegenerate.
For radial one-mode estimation, ideal heterodyne detection accesses the
exact fraction \((\nu-\hbar/2)/(\nu+\hbar/2)\) of the SLD quantum Fisher
information.  Its vanishing boundary limit reflects finite heterodyne
information relative to a divergent radial quantum Fisher information, not
zero measurement information.  Finally, a minimal Schur realization fixes
the auxiliary inertia and identifies the second Schur complement as the
true covariant vertical block.
\end{abstract}

\section{Introduction}

Centered classical Gaussian measures and centered bosonic Gaussian states are
both parametrized by covariance matrices, but their information geometries
have different meanings.  The former carry the covariance Fisher--Rao
metric of a classical statistical model \cite{Amari2016}.  The latter obey
the Robertson--Schr\"odinger condition and inherit the Bures metric from
quantum fidelity \cite{Bures1969,Uhlmann1976,Braunstein1994,Banchi2015}.
The quantum geometry is therefore not obtained merely by restricting the
classical positive cone.

Gaussian phase-space methods make the symmetric logarithmic derivative and
quantum Fisher information explicit
\cite{Pinel2012,Monras2013,Safranek2019,HuangWilde2026}.  We ask a
complementary structural question: what is the exact contravariant
difference between the classical covariance Fisher tensor and the quantum
Bures tensor?  The answer is particularly simple.  The Fisher cometric is
quadratic in \(\Sigma\), whereas the Bures cometric contains one additional
term fixed by the canonical commutation relations.  Under the linear map
\(P\mapsto\Omega P\), this term is the trace form of
\(\mathfrak{sp}(2N,\RR)\), and therefore a fixed multiple of its Killing
form.

At the covariant level, an additive even--odd splitting of the Gaussian
quantum Fisher information was recently introduced by Chatterjee
\emph{et al.}\ \cite{Chatterjee2026}, separating symplectic-spectrum
changes from frame-deforming dynamics.  The present work takes a
complementary, contravariant approach: the exact difference between the
classical covariance Fisher tensor and the quantum Bures tensor is
identified as a state-independent Killing correction
(Theorem~\ref{thm:killing-form-identity}), whose consequences --- Petz
uniqueness, boundary nullity, heterodyne accessibility, and the minimal
Schur realization --- are developed here.

The Killing identification has three consequences.  First, symplectic invariance makes
the form of every state-independent correction unique up to scale.  More
strongly, within the normalized symmetric Petz family, requiring the
Gaussian covariance cometric to equal a lift-normalized Fisher block plus a
state-independent invariant term selects Bures uniquely.  Second, the
Cartan decomposition fixes the signs of the correction on compact and
noncompact covariance sectors.  These sectors acquire their familiar
passive--active optical meaning only after a physical tangent encoding has
been specified.  Third, the exact Williamson-frame spectrum determines the
boundary stratification and its stabilizer kernel.

The metrological interpretation requires an additional distinction.  The
Killing term is an intrinsic, indefinite correction on covariance
covectors.  The deficit between SLD quantum Fisher information and the
classical Fisher information of a chosen POVM is instead nonnegative,
measurement-dependent, and acts on parameter tangents.  We keep these two
notions separate and use radial heterodyne estimation only as an explicit
comparison between intrinsic and measurement-accessible information.

This algebraic rigidity carries a direct physical consequence. Quantum mechanics does not merely emerge as an arbitrary dynamical correction to classical statistics. Rather, if one requires an information geometry that symmetrically incorporates both macroscopic probability and invariant phase-space mechanics, the mathematics forces a single, unique completion within the Petz family. The symplectic Killing correction dictates that Bures is the canonical geometric bridge between classical statistics and quantum phase space.

Section~\ref{sec:conventions} derives the covariance tensors, compares their algebraic hierarchy, and proves the Killing characterization. Section~\ref{sec:uniqueness} establishes Petz uniqueness. Section~\ref{sec:metrology} separates the intrinsic gap from measurement deficits. Section~\ref{sec:boundary} solves the boundary spectrum, and Section~\ref{sec:classicality} evaluates heterodyne accessibility. Section~\ref{sec:lift} gives the compact Schur construction and states the scope and curvature outlook.

\paragraph*{Relation to a previous version and to companion work.}
A letter-length precursor of this article announced the Killing-form
identity and the minimal Schur construction
(Theorems~\ref{thm:killing-form-identity} and~\ref{thm:minimal-lift})
and is superseded here.\footnote{arXiv:2607.16376v1.  The present
article replaces that record.}  Beyond consolidating those results,
the following are new: the uniqueness of Bures within the normalized
symmetric Petz family (Theorem~\ref{thm:petz-bures-selection}) together
with the BKM boundary diagnostic; the separation of the intrinsic
Killing correction from measurement-induced information deficits
(Section~\ref{sec:metrology}) and the exact heterodyne and
photon-number comparisons (Section~\ref{sec:classicality}); and the
covariant vertical-block theorem with its frame-dependence
counterexample (Theorem~\ref{thm:vertical-block},
Appendix~\ref{app:schur}), which corrects the naive reading of the
fiber block in the earlier version.
The multimode Williamson-frame spectrum and its boundary kernel
(Section~\ref{sec:boundary}) are established here for the first time
in this form.

\section{Gaussian covariance geometry and Killing characterization}
\label{sec:conventions}

Let
\begin{equation}
  \Om=\bigoplus_{k=1}^N
  \begin{pmatrix}0&1\\-1&0\end{pmatrix},
  \qquad
  [\widehat R_a,\widehat R_b]=i\hbar\Om_{ab}.
\end{equation}
For a centered state, set
\begin{equation}
  \Sigma_{ab}=\frac12\left\langle
  \widehat R_a\widehat R_b+\widehat R_b\widehat R_a\right\rangle.
\end{equation}
The faithful centered Gaussian domain is
\begin{equation}
  \mathcal M_Q^\circ=
  \left\{\Sigma\in\Sym_{++}(2N):
  \Sigma+\frac{i\hbar}{2}\Om>0\right\}.
  \label{eq:quantumdomain}
\end{equation}
Williamson's theorem \cite{Williamson1936,deGosson2006,Weedbrook2012} gives
\begin{equation}
  \Sigma=SDS^{\mathsf T},\qquad
  D=\bigoplus_{k=1}^N\nu_k\Id_2,
  \qquad S\Om S^{\mathsf T}=\Om,
  \label{eq:williamson}
\end{equation}
and Eq.~\eqref{eq:quantumdomain} is equivalent to \(\nu_k>\hbar/2\) for every \(k\). The closed domain allows equality and is stratified by the set of saturated Williamson modes.

We identify covariance tangents and cotangents with \(\Sym(2N)\) using the trace pairing \(\langle P,Q\rangle_{\rm tr}=\Tr(PQ)\).

For a centered classical Gaussian density, differentiation of the log likelihood and evaluation of the Gaussian fourth moment give the covariance Fisher metric
\begin{equation}
  g_{{\rm FR},\Sigma}(X,Y)
  =\frac12\Tr(\Sigma^{-1}X\Sigma^{-1}Y).
  \label{eq:fishermetric}
\end{equation}
Its sharp map is \(P\mapsto2\Sigma P\Sigma\), so
\begin{equation}
  g_{{\rm FR},\Sigma}^*(P,Q)=2\Tr(P\Sigma Q\Sigma).
\end{equation}
We compare Bures with the lift-normalized Fisher cometric
\begin{equation}
  \mathcal F_\Sigma^*(P,Q)=8\Tr(P\Sigma Q\Sigma).
  \label{eq:fisher}
\end{equation}
Thus \(\mathcal F^*=4g_{\rm FR}^*\); the factor four is exactly the high-noise normalization implied by the convention \(g_{\rm B}=F_{\rm SLD}/4\).

For a covariance tangent \(X\), the centered Gaussian SLD has a real symmetric quadratic gain \(\mathfrak G_X\) satisfying \cite{Monras2013,Safranek2019,HuangWilde2026}
\begin{equation}
  \Sigma\mathfrak G_X\Sigma
  +\frac{\hbar^2}{4}\Om\mathfrak G_X\Om=X,
  \qquad
  F_{\rm SLD}(X,Y)=\frac12\Tr(\mathfrak G_XY).
  \label{eq:sldgain}
\end{equation}
Introduce the self-adjoint superoperator \(\mathcal S_\Sigma(G)= \Sigma G\Sigma+\frac{\hbar^2}{4}\Om G\Om\). On the faithful domain it is invertible on \(\Sym(2N)\), and
\begin{equation}
  g_{{\rm B},\Sigma}(X,Y)
  =\frac18\Tr\!\left[\mathcal S_\Sigma^{-1}(X)Y\right].
  \label{eq:buresmetric}
\end{equation}
The Bures sharp map is therefore \(P\mapsto8\mathcal S_\Sigma(P)\). In the trace pairing, the Gaussian Bures cometric is
\begin{equation}
  g_{{\rm B},\Sigma}^*(P,Q)
  =8\Tr(P\Sigma Q\Sigma)
   +2\hbar^2\Tr(P\Om Q\Om).
  \label{eq:bures}
\end{equation}
The second term is indefinite on the ambient vector space, while their sum is positive definite on the faithful quantum domain. Notice that the plus sign in Eq.~\eqref{eq:bures} is essential: for a scalar covector, \(\Tr(P\Om P\Om)<0\).

\begin{remark}[Literature comparison and algebraic hierarchy]
The exact algebraic forms of these covariance cometrics are dispersed across the optimal transport, classical statistics, and quantum metrology literature \cite{Bhatia2019,Malago2018,Amari2016,Pinel2012}. Translated into a unified trace-pairing convention, the classical Bures--Wasserstein, classical Fisher--Rao, and quantum Bures cometrics evaluate to:
\begin{align}
  g_{\rm BW}^*(P,P) &= 4\Tr(P\Sigma P), \label{eq:lit-bw} \\
  g_{\rm FR}^*(P,P) &= 2\Tr(P\Sigma P\Sigma), \label{eq:lit-fr} \\
  g_{\rm B}^*(P,P) &= 8\Tr(P\Sigma P\Sigma) + 2\hbar^2\Tr(P\Om P\Om). \label{eq:lit-bures}
\end{align}
Placed side-by-side, a strict algebraic hierarchy becomes visible. The macroscopic optimal-transport cometric is linear in the state \(\Sigma\), while the classical Fisher geometry is strictly quadratic. The quantum Bures geometry is the only one that carries an anomalous, state-independent constant term (\(O(\Sigma^0)\)). The central observation of this paper is that this constant term is not an arbitrary metric deformation, but exactly the invariant Killing form of the phase space.
\end{remark}

\subsection{Cartan decomposition and the invariant Killing correction}
\label{subsec:identity}

\begin{lemma}[Symplectic intertwiner]
The map
\begin{equation}
  \iota:\Sym(2N)\longrightarrow\mathfrak{sp}(2N,\RR),
  \qquad \iota(P)=\Om P,
  \label{eq:iota}
\end{equation}
is a linear isomorphism. If \(\mathcal I(P)=\Om P\Om^{\mathsf T}\), then
its \(+1\) eigenspace maps to the compact factor
\(\mathfrak k\cong\mathfrak u(N)\), while its \(-1\) eigenspace maps to
the noncompact factor \(\mathfrak p\). Their dimensions are \(N^2\) and
\(N^2+N\), respectively.
\end{lemma}

\begin{proof}
For symmetric \(P\), one has \((\Om P)^{\mathsf T}\Om+\Om(\Om P)=0\), so \(\Om P\) is symplectic. Injectivity and equality of dimensions give bijectivity. If \(\mathcal I(P)=P\), then \(P\) commutes with \(\Om\) and \(\Om P\) is antisymmetric, hence compact. If \(\mathcal I(P)=-P\), then \(\Om P\) is symmetric, hence noncompact.
\end{proof}

Let \(\tau(X,Y)=\Tr(XY)\) on \(\mathfrak{sp}(2N,\RR)\). The Killing
form is \(B_{\rm Kill}(X,Y)=2(N+1)\tau(X,Y)\). It is negative definite
on \(\mathfrak k\) and positive definite on \(\mathfrak p\)
\cite{Helgason1978}.

\begin{theorem}[Killing-form identity]
\label{thm:killing-form-identity}
For every faithful centered Gaussian covariance matrix,
\begin{equation}
  g_{{\rm B},\Sigma}^*(P,Q)
  =\mathcal F_\Sigma^*(P,Q)
   +2\hbar^2\tau\bigl(\iota(P),\iota(Q)\bigr).
  \label{eq:killing}
\end{equation}
Thus the compact covariance sector receives a negative cometric correction and the noncompact sector a positive correction.
\end{theorem}

\begin{proof}
Cyclicity of the trace gives \(\tau(\Om P,\Om Q)=\Tr(P\Om Q\Om)\). Substitution in Eq.~\eqref{eq:bures} proves the identity. The signs follow from the Cartan decomposition in the preceding lemma.
\end{proof}

\begin{theorem}[Invariant uniqueness]
\label{thm:invariant-uniqueness}
Suppose that a state-independent symmetric bilinear form \(C\) on \(\Sym(2N)\) is invariant under the cotangent action \(P\longmapsto S^{-\mathsf T}PS^{-1}\) for every \(S\in\Sp(2N,\RR)\). Then there is a constant \(c\) such that
\begin{equation}
  C(P,Q)=c\,\tau\bigl(\iota(P),\iota(Q)\bigr).
  \label{eq:unique-invariant}
\end{equation}
In particular, any symplectically covariant Gaussian cometric that differs from \(\mathcal F^*\) by a state-independent term has a Killing correction.
\end{theorem}

\begin{proof}
Under \(\iota\), the cotangent action is intertwined with the adjoint action of \(\Sp(2N,\RR)\). Hence \(C\) induces an invariant symmetric bilinear form on \(\mathfrak{sp}(2N,\RR)\). The real symplectic algebra is absolutely simple (its complexification is simple), so by Schur's Lemma every such invariant form is proportional to its Killing form \cite{Helgason1978}, equivalently to \(\tau\).
\end{proof}

The canonical commutation relations fix the specific scale of this correction to \(2\hbar^2\).

The Cartan signs acquire an optical interpretation only after a tangent
model has been specified.  For a Gaussian-unitary encoding
\(X=Y\Sigma+\Sigma Y^{\mathsf T}\), compact generators organize passive
rotations and noncompact generators organize active squeezing.  The signs in
Eq.~\eqref{eq:killing} concern the correction of the dual tensor; they do not
make the total Bures metric indefinite and do not impose a universal ordering
of gate costs or estimation precision.

\begin{remark}[The geometric role of Planck's constant]
While \(\hbar\) has a fixed empirical value in quantum physics, structurally it acts as a symplectic scaling parameter. In Eq.~\eqref{eq:killing}, it calibrates the relative weight of the non-commutative Killing correction against the macroscopic Fisher background, and thereby governs the geometric distance to the pure-state boundary.
\end{remark}

\section{Uniqueness within the symmetric Petz family}
\label{sec:uniqueness}

Quantum monotonicity alone does not select a unique Riemannian metric on
density operators, but rather the Petz family
\cite{Petz1996,PetzSudar1996}.  We now ask whether the additive splitting in
Eq.~\eqref{eq:killing} selects Bures within that family.

Consider the one-mode isotropic family \(\Sigma_\nu=\nu\Id_2\). Set
\begin{equation}
 x=\frac{2\nu}{\hbar}>1,
 \quad
 E_0=\frac{\Id_2}{\sqrt2},
 \quad
 E_+=\frac{1}{\sqrt2}\begin{pmatrix}1&0\\0&-1\end{pmatrix},
 \quad
 E_\times=\frac{1}{\sqrt2}\begin{pmatrix}0&1\\1&0\end{pmatrix}.
 \label{eq:auxiliary-basis}
\end{equation}
A normalized symmetric operator-monotone function obeys \(f(1)=1\) and \(f(t)=tf(t^{-1})\). It defines a Petz information metric with spectral form
\begin{equation}
 I_{f,\rho}(A,A)
 =\sum_{m,n}\frac{|A_{mn}|^2}{p_n f(p_m/p_n)},
 \qquad
 g_f=\frac14 I_f.
 \label{eq:petz-spectral}
\end{equation}
For the geometric thermal spectrum, the spectral sums in Eq.~\eqref{eq:petz-spectral} converge absolutely for any operator-monotone \(f\). The Bures convention corresponds to \(f_{\rm B}(t)=(1+t)/2\). For \(\Sigma_\nu\), the thermal number distribution is geometric, \(p_n=(1-q)q^n\), with \(q=(x-1)/(x+1)\).

\begin{proposition}[Local Gaussian reduction of the Petz family]
\label{prop:petz-reduction}
Define \(a_f(x)=f(q^2)/f_{\rm B}(q^2)\). In the basis of Eq.~\eqref{eq:auxiliary-basis}, the covariance cometric of \(g_f\) has eigenvalues
\begin{align}
 \lambda_{\rm rad}^{f,*}
 &=2\hbar^2(x^2-1),
 \label{eq:petz-radial}\\
 \lambda_{\rm shape}^{f,*}
 &=2\hbar^2(x^2+1)a_f(x)
 \quad\text{(multiplicity two)}.
 \label{eq:petz-shape}
\end{align}
\end{proposition}

\begin{proof}
Radial variations preserve the number basis and sample only \(f(1)=1\). The two traceless covariance directions are generated by quadratic squeezing operators, connecting only \(n\) with \(n\pm2\). Since \(p_{n+2}/p_n=q^2\), every nonzero term receives the same \(f\)-dependent factor. Relative to Bures, the covariant shape weight is multiplied by \(f_{\rm B}(q^2)/f(q^2)\). Because the pullback of the metric onto this three-dimensional tangent subspace is diagonal, its matrix inverse gives the contravariant eigenvalues of the dual metric in Eq.~\eqref{eq:petz-shape}. 
\end{proof}

\begin{theorem}[Bures uniqueness within the Petz family]
\label{thm:petz-bures-selection}
Among normalized symmetric Petz metrics, Bures is the unique metric whose
one-mode Gaussian covariance cometric equals the lift-normalized Fisher
cometric \(\mathcal F^*\) plus a state-independent, symplectically invariant
correction.
\end{theorem}

\begin{proof}
At \(\Sigma_\nu\), the Fisher cometric has the common eigenvalue \(2\hbar^2x^2\). The radial difference in Eq.~\eqref{eq:petz-radial} is \(-2\hbar^2\). By Theorem~\ref{thm:invariant-uniqueness}, any state-independent symplectically invariant symmetric bilinear form is proportional to the Killing form. Thus, it must take the opposite value \(+2\hbar^2\) in both shape directions. Equation~\eqref{eq:petz-shape} therefore forces \(a_f(x)=1\) for every \(x>1\). Since \(q^2\) covers \((0,1)\), one has \(f=f_{\rm B}\) there; symmetry extends the equality to \((1,\infty)\), and continuity covers \(t=1\). Because a Petz metric is determined entirely by its operator-monotone function, \(f=f_{\rm B}\) forces \(g_f=g_{\rm B}\) on all states, not merely the Gaussian family. The converse is exactly Theorem \ref{thm:killing-form-identity}.
\end{proof}

This statement is confined to the normalized symmetric Petz family.  As a
useful boundary diagnostic, the normalized BKM metric
\(f_{\rm BKM}(t)=(t-1)/\log t\) has shape eigenvalue
\begin{equation}
 \widehat\lambda_{\rm shape}^{\rm BKM}
 =4\hbar^2\frac{x}{\log[(x+1)/(x-1)]} \longrightarrow0
 \qquad(x\downarrow1).
 \label{eq:bkm-boundary}
\end{equation}
Thus regular metrics such as Bures retain a finite pure-orbit cometric,
whereas the nonregular BKM metric develops an additional logarithmic boundary
degeneracy.  This illustrates the force of the invariant-correction
requirement without promoting it to a uniqueness claim outside the stated
Petz class.

\section{Metrological bounds and the intrinsic covariance gap}
\label{sec:metrology}

In quantum metrology, the Bures metric determines the ultimate precision limit for parameter estimation via the Quantum Cram\'er--Rao Bound (QCRB). For a one-parameter model \(\theta\mapsto\Sigma_\theta\) with covariance tangent \(X_\theta=\partial_\theta\Sigma_\theta\), the SLD quantum Fisher information is \(F_Q(\theta) = 4g_{{\rm B},\Sigma_\theta}(X_\theta,X_\theta)\).

A measurement must be specified before classical outcome statistics exist. For a centered Gaussian POVM \(\mathsf M=\{M_y\}\), the classical Fisher information \(I_{\mathsf M}(\theta)\) bounds the estimator variance. The Braunstein--Caves characterization guarantees \(I_{\mathsf M}(\theta)\leq F_Q(\theta)\). For a centered Gaussian measurement yielding a zero-mean classical output covariance \(V_{\mathsf M}(\Sigma)=A\Sigma A^{\mathsf T}+N_{\mathsf M}\), the outcome Fisher information is purely the pullback of the classical covariance Fisher metric:
\begin{equation}
 I_{\mathsf M}(\theta)
 =\frac12\Tr\!\left[
 V_{\mathsf M}^{-1}\dot V_{\mathsf M}
 V_{\mathsf M}^{-1}\dot V_{\mathsf M}
 \right].
 \label{eq:gaussian-povm-fisher}
\end{equation}

\begin{proposition}[Intrinsic and measurement-induced gaps are distinct]
\label{prop:two-gaps}
The Killing identity
\begin{equation}
 g_{{\rm B},\Sigma}^*-\mathcal F_\Sigma^*
 =2\hbar^2\,\tau\circ(\iota\times\iota)
 \label{eq:intrinsic-gap}
\end{equation}
is an equality between intrinsic cometrics on covariance covectors. By contrast, for a fixed POVM the measurement deficit \(\Delta_{\mathsf M}(X) := F_Q(X)-I_{\mathsf M}(X)\) is a nonnegative quadratic form on parameter tangents and depends on the chosen apparatus. No measurement-independent identification of \(\Delta_{\mathsf M}\) with the Killing correction is possible in general.
\end{proposition}

\begin{proof}
The first statement is Eq.~\eqref{eq:killing}.  The Braunstein--Caves
inequality gives \(\Delta_{\mathsf M}\geq0\).  In contrast, the Killing
form is negative on \(\mathfrak k\) and positive on \(\mathfrak p\).
Moreover, the two forms have different tensor type and the latter requires a
POVM pullback.  These facts exclude a universal equality.
\end{proof}

Accordingly, Cartan type alone does not impose a universal precision ordering
between rotations and squeezing.  Such an ordering depends on the encoded
tangent, the state, and the measurement. 

\section{Williamson spectrum and boundary stratification}
\label{sec:boundary}

\begin{theorem}[Multimode Williamson-frame spectrum]
\label{thm:multimode-spectrum}
At the Williamson representative \(D=\bigoplus_k\nu_k\Id_2\), with \(x_k=2\nu_k/\hbar\), the Bures covariance cometric decomposes into the following sectors:
\begin{center}
\begin{tabular}{lll}
\toprule
sector & weight & multiplicity\\
\midrule
local scalar \(k\) & \(2\hbar^2(x_k^2-1)\) & \(N\) total\\
local traceless \(k\) & \(2\hbar^2(x_k^2+1)\) & \(2N\) total\\
negative-parity pair \(k<\ell\) & \(2\hbar^2(x_kx_\ell-1)\) & \(2\) per pair\\
positive-parity pair \(k<\ell\) & \(2\hbar^2(x_kx_\ell+1)\) & \(2\) per pair\\
\bottomrule
\end{tabular}
\end{center}
The multiplicities sum to \(N(2N+1)\).
\end{theorem}

\begin{proof}
For a diagonal block, the involution \(P\mapsto\Om P\Om\) has eigenvalue
\(-1\) on the scalar matrix and \(+1\) on the two traceless symmetric
matrices.  For a cross block \(A\in M_2(\RR)\), write
\[
 J=\begin{pmatrix}0&1\\-1&0\end{pmatrix},\quad
 X=\begin{pmatrix}0&1\\1&0\end{pmatrix},\quad
 Z=\begin{pmatrix}1&0\\0&-1\end{pmatrix}.
\]
The involution \(A\mapsto JAJ\) has negative eigenspace
\(\operatorname{span}\{\Id_2,J\}\) and positive eigenspace
\(\operatorname{span}\{X,Z\}\), both two-dimensional. Substitution into
Eq.~\eqref{eq:bures} gives the four weights. General \(\Sigma\) follows by
Williamson congruence.
\end{proof}

\begin{corollary}[Boundary nullity and stabilizer]
\label{cor:nullity}
Suppose exactly \(m\) Williamson modes satisfy \(x_k=1\) (pure modes). Then
\begin{equation}
  \dim\ker g_{{\rm B},\Sigma}^*=m^2,
  \qquad
  \ker g_{{\rm B},\Sigma}^*\cong\mathfrak u(m).
  \label{eq:nullity}
\end{equation}
The complementary pure-orbit sector has dimension \(m(m+1)\), equal to \(\dim[\Sp(2m,\RR)/U(m)]\).
\end{corollary}

\begin{proof}
The kernel contains one scalar direction for each pure mode and two
negative-parity directions for each pair of pure modes, giving
\(m+2\binom m2=m^2\). On the pure block,
the Bures sharp operator is
\(\mathcal G_{\rm B}(P)=2\hbar^2(P+\Om_mP\Om_m)\), so a symmetric \(P\)
is null precisely when it commutes with \(\Om_m\). The map
\(P\mapsto\Om_mP\) identifies this symmetric commutant with
\(\mathfrak{sp}(2m,\RR)\cap\mathfrak{so}(2m)=\mathfrak u(m)\).
\end{proof}

At a rank-changing boundary, the continuous interior SLD tensor must be
distinguished from the metric intrinsic to the fixed-rank stratum
\cite{Safranek2017}. The kernel above describes the smooth algebraic extension of the dual polynomial \(g_{{\rm B},\Sigma}^*\) to the closed domain boundary, which geometrically removes spectrum-changing directions from the intrinsic pure-state stratum; it does not make the pure Gaussian
orbit blind.  On \(\Sp(2m,\RR)/U(m)\), Bures reduces, up to convention, to
the nondegenerate Fubini--Study metric \cite{Bengtsson2017}.

\begin{remark}[Relation to the even--odd splitting of the Gaussian QFI]
\label{rem:chatterjee}
Chatterjee \emph{et al.}~\cite{Chatterjee2026} decompose the QFI of
centered multimode Gaussian states additively by \(\Om\)-parity of the
covariant tangent, using the same Cartan decomposition of
\(\mathfrak{sp}(2N,\RR)\): the spectrum-changing parity sector vanishes
identically on the pure-state manifold, while the complementary sector
restricts there to the Siegel upper half-space metric.
At a Williamson frame, the Bures sharp map preserves \(\Om\)-parity, so
the sector decomposition of Theorem~\ref{thm:multimode-spectrum} is the
contravariant image of their covariant splitting, and the nondegenerate
pure-orbit metric on \(\Sp(2m,\RR)/U(m)\) in
Corollary~\ref{cor:nullity} is consistent with their Siegel-metric
identification.  The two results are nevertheless of different type.
The even--odd splitting is a state-dependent orthogonal decomposition
of the covariant quadratic form on parameter tangents.  The present
characterization concerns the contravariant difference from the
classical Fisher cometric, which is state-independent and, by the
invariant-uniqueness theorem, algebraically rigid; no covariant
analogue exists, since the covariant difference
\(g_{\rm B}-g_{\rm FR}\) is state-dependent and of no fixed invariant
type.  The Petz selection of Theorem~\ref{thm:petz-bures-selection}
and the Schur inertia realization of Section~\ref{sec:lift} likewise
have no counterpart in the covariant splitting.
The intermediate strata are also resolved differently: for \(m<N\) pure modes, Corollary~\ref{cor:nullity} identifies the kernel as the \(m^2\)-dimensional stabilizer algebra \(\mathfrak u(m)\), refining the all-modes-pure statement to arbitrary partial purity.
\end{remark}

\section{Classical limit and measurement-induced statistics}
\label{sec:classicality}

We next compare the intrinsic classical limit with the statistics generated
by one particular detector.  This comparison does not identify the Killing
correction with a POVM deficit.

\begin{proposition}[Intrinsic high-noise suppression]
\label{prop:high-noise}
On a sector on which \(\mathcal F_\Sigma^*\) is uniformly nonzero,
\begin{equation}
  \frac{\left|g_{{\rm B},\Sigma}^*(P,P)
  -\mathcal F_\Sigma^*(P,P)\right|}
  {\mathcal F_\Sigma^*(P,P)}
  =O\!\left(\frac{\hbar^2}{\nu_{\min}^2}\right)
  \longrightarrow0
  \qquad
  \left(\frac{\nu_{\min}}{\hbar}\to\infty\right).
  \label{eq:high-noise-suppression}
\end{equation}
\end{proposition}

\begin{proof}
The correction in Eq.~\eqref{eq:killing} is independent of the scale of
\(\Sigma\) and is bounded by a constant times
\(\hbar^2\lVert P\rVert_F^2\).  On a uniformly nondegenerate Fisher sector,
Eq.~\eqref{eq:fisher} is bounded below by a constant times
\(\nu_{\min}^2\lVert P\rVert_F^2\).  Their quotient gives the result.
\end{proof}

Thus the classical covariance Fisher geometry is recovered asymptotically
at high noise.  This is an intrinsic, measurement-independent statement.

To see how this manifests in the laboratory, consider the extraction of the radial (purity) parameter \(\nu\) via an ideal heterodyne measurement. Set \(a=\hbar/2\) and \(\Sigma_\nu=\nu\Id_2\) for \(\nu>a\). In standard quadrature normalization, ideal heterodyne detection adds one vacuum unit to the simultaneous phase-space readout, yielding the classical output covariance \(V_{\rm het}(\nu)=(\nu+a)\Id_2\). Its classical Fisher information derived from Eq.~\eqref{eq:gaussian-povm-fisher} is \(I_{\rm het}(\nu)=(\nu+a)^{-2}\). 

The fundamental radial SLD quantum Fisher information, derived from the exact gain equations (see Appendix~\ref{app:onemode}), is \(F_Q(\nu)=(\nu^2-a^2)^{-1}\). 

\begin{proposition}[Exact heterodyne accessibility ratio]
\label{prop:heterodyne-ratio}
For the faithful radial family, the exact fraction of the fundamental quantum Fisher information accessible to ideal heterodyne detection is
\begin{equation}
 \frac{I_{\rm het}(\nu)}{F_Q(\nu)}
 =\frac{\nu-a}{\nu+a}.
 \label{eq:heterodyne-ratio}
\end{equation}
\end{proposition}

\begin{proof}
For the two-dimensional Gaussian output,
\(\dot V_{\rm het}=\Id_2\), so
Eq.~\eqref{eq:gaussian-povm-fisher} gives \(I_{\rm het}\).  Division by
the radial QFI gives Eq.~\eqref{eq:heterodyne-ratio}.
\end{proof}

The high-noise limit of Eq.~\eqref{eq:heterodyne-ratio} realizes the
classical covariance limit for this parameter and detector.  At the pure
boundary the ratio tends to zero for a different reason: the heterodyne
information remains finite, \(I_{\rm het}(a)=1/(4a^2)\), whereas the
faithful-state radial QFI diverges.  Hence heterodyne detection does not
inherit the \(m^2\)-dimensional cometric kernel as vanishing classical Fisher
information.

Nor does the kernel imply a universal symmetry-breaking requirement for
measurements.  Indeed, the radial thermal probabilities
\(p_n=(1-q)q^n\), with \(q=(\nu-a)/(\nu+a)\), give the phase-invariant
photon-number Fisher information
\(I_{\rm num}=1/(\nu^2-a^2)=F_Q\).  Measurement accessibility depends
jointly on the tangent encoding and the POVM; breaking the stabilizer
symmetry is neither generally necessary nor sufficient.

\section{Minimal Schur realization and scope}
\label{sec:lift}

Equation \eqref{eq:killing} presents a conceptual puzzle: the Bures cometric is strictly positive definite on the faithful domain, yet it additively incorporates the Killing form, which is indefinite (negative on the compact sector (\(\mathfrak{k}\)) and positive on the noncompact sector (\(\mathfrak{p}\))). Geometrically, such indefinite corrections arise naturally when a positive base metric is formed as a quotient or by 'integrating out' an auxiliary space. We now construct the exact algebraic realization of this process via block elimination.  No
physical purification dynamics or Riemannian quotient interpretation is
assumed.

Let \(V=\Sym(2N)\),
\(\mathfrak g=\mathfrak{sp}(2N,\RR)\), and define
\begin{equation}
 K_f=-\tau,\qquad
 \iota_\hbar(P)=\frac{\hbar}{2}\Om P,\qquad
 R(P,\pi)=2\sqrt2\,K_f\bigl(\iota_\hbar(P),\pi\bigr).
\end{equation}
On \(V\oplus\mathfrak g\), set
\begin{equation}
 G^*=\begin{pmatrix}
 \mathcal F_\Sigma^*&R\\
 R^{\mathsf T}&K_f
 \end{pmatrix}.
 \label{eq:lift}
\end{equation}

\begin{theorem}[Minimal Schur realization]
\label{thm:minimal-lift}
The Schur complement of \(K_f\) is the Bures cometric,
\begin{equation}
 \mathcal F_\Sigma^*-RK_f^{-1}R^{\mathsf T}
 =g_{{\rm B},\Sigma}^*.
 \label{eq:fiberschur}
\end{equation}
Every nondegenerate Schur realization of the full Killing correction on an
auxiliary space \(W\) satisfies \(\dim W\geq N(2N+1)\). At equality its
auxiliary form \(\widetilde K\) is fixed up to congruence and has
\begin{equation}
 \operatorname{In}(\widetilde K)=(N^2,N^2+N),
 \label{eq:fiber-inertia}
\end{equation}
where the entries count positive and negative dimensions.  The canonical
construction above realizes this minimum with \(\widetilde K=K_f\).
\end{theorem}

\begin{proof}
The definition of \(R\) gives
\begin{equation}
 (RK_f^{-1}R^{\mathsf T})(P,Q)
 =8K_f\bigl(\iota_\hbar(P),\iota_\hbar(Q)\bigr)
 =-2\hbar^2\tau\bigl(\iota(P),\iota(Q)\bigr),
\end{equation}
which proves Eq.~\eqref{eq:fiberschur} via Eq.~\eqref{eq:killing}.  For an
alternative realization on an auxiliary space \(W\), the correction factors
as \(C=-\widetilde R\widetilde K^{-1}\widetilde R^{\mathsf T}\).  Hence
\(\operatorname{rank}C\leq\dim W\).  Nondegeneracy of the Killing form gives
\(\operatorname{rank}C=\dim V=N(2N+1)\).  At equality,
\(\widetilde R\) is invertible, so Sylvester's law of inertia fixes
\(\operatorname{In}(\widetilde K)=(N^2,N^2+N)\), the inertia of
\(K_f=-\tau\).
\end{proof}

\begin{theorem}[Correct covariant vertical block]
\label{thm:vertical-block}
Let \(G=(G^*)^{-1}\).  The restriction of \(G\) to vertical coordinate
vectors is governed by the second Schur complement,
\begin{equation}
 G_{\rm vert}
 =\left(K_f-R^{\mathsf T}(\mathcal F_\Sigma^*)^{-1}R\right)^{-1},
 \label{eq:vertical}
\end{equation}
not by \(K_f^{-1}\).  Throughout the faithful domain,
\begin{equation}
 \operatorname{In}(G_{\rm vert})=(N^2,N^2+N).
 \label{eq:vertical-inertia}
\end{equation}
Only this inertia is fixed globally.  The positive and negative subspaces of
\(G_{\rm vert}\) need not coincide with the fixed Cartan subspaces
\(\mathfrak k\) and \(\mathfrak p\) away from a Williamson-adapted frame.
\end{theorem}

\begin{proof}
Block inversion with the base block inverted first gives
Eq.~\eqref{eq:vertical}.  Put
\(S_f=K_f-R^{\mathsf T}(\mathcal F_\Sigma^*)^{-1}R\).  Haynsworth inertia
additivity applied once through the fiber block and once through the base
block gives
\begin{align}
 \operatorname{In}(G^*)&=\operatorname{In}(K_f)
  +\operatorname{In}(g_{{\rm B},\Sigma}^*),\\
 \operatorname{In}(G^*)&=\operatorname{In}(\mathcal F_\Sigma^*)
  +\operatorname{In}(S_f).
\end{align}
Because \(g_{{\rm B},\Sigma}^*\) and \(\mathcal F_\Sigma^*\) are positive
definite and have the same dimension,
\(\operatorname{In}(S_f)=\operatorname{In}(K_f)\).  Inversion preserves
inertia \cite{Haynsworth1968,Zhang2005}.  The one-mode calculation in
Appendix~\ref{app:schur} proves that this statement does not fix the Cartan
sign subspaces.
\end{proof}

The construction is therefore an exact algebraic block realization.  It
identifies the minimal auxiliary dimension and inertia required to reproduce
the Killing correction and determines the true covariant vertical block.
By itself, it does not identify the auxiliary variables with environmental
modes, a purification Hilbert space, or physical channel degrees of freedom.
Such an interpretation would require an additional Gaussian dilation or
control model.

\subsection{Scope and curvature outlook}
\label{subsec:scope}

The covariance geometry fixes an admissible domain, invariant sector weights,
and boundary nullities; it does not choose a measurement or a state-update
law.  It also does not diagnose non-Gaussianity, since states with different
higher moments may share the same covariance.  Likewise, the indefinite
Killing block is not by itself a positive Nielsen circuit cost.  Circuit
penalties, purification control, and dynamical thermodynamics require
additional structures and are left to companion work.

The Killing characterization does, however, organize a natural curvature
program.  The full Gaussian Bures metric is not homogeneous because its
coefficients depend on the Williamson spectrum.  On each fixed-spectrum
symplectic orbit, the restricted metric is invariant, so its connection and
curvature can be organized using homogeneous-space formulas and the Lie
brackets of \(\mathfrak k\) and \(\mathfrak p\).  The ambient calculation
must additionally include spectral derivatives and radial--orbit coupling.
The present results provide the algebraic foundation for that calculation,
but do not replace it.

\section{Conclusion}

The Gaussian Bures covariance cometric is the lift-normalized Fisher
cometric plus a state-independent Killing term.  Symplectic invariance fixes
the algebraic form up to scale, while the canonical commutation relations
fix its Bures coefficient.  Within the normalized symmetric Petz family,
requiring this state-independent invariant correction selects Bures
uniquely.

The operational layer is distinct.  The Killing correction is an intrinsic
indefinite tensor on covariance covectors, not the positive information
deficit of a chosen measurement.  Ideal heterodyne detection nevertheless
provides an exact one-mode comparison: it becomes asymptotically efficient
at high noise, while its boundary ratio vanishes because its information is
finite relative to a divergent radial QFI.  The Williamson spectrum supplies
the complementary global result: \(m\) pure modes produce an
\(m^2\)-dimensional cometric kernel isomorphic to \(\mathfrak u(m)\), while
the pure Gaussian orbit remains nondegenerate.  Together, these results demonstrate that Gaussian quantum geometry is not an arbitrary modification of classical Fisher statistics, but the unique Petz geometry selected by invariant phase-space structure.

\appendix

\section{Derivation of the one-mode QFI}
\label{app:onemode}

By covariance under fixed symplectic changes of frame, write the one-mode
state as
\(\Sigma(\nu,r,\alpha)=\nu R(\alpha)D R(\alpha)^{\mathsf T}\), where
\(D=\operatorname{diag}(e^{2r},e^{-2r})\) and \(R(\alpha)\) is the
two-dimensional rotation matrix. Setting \(\alpha=0\) gives
\(\Sigma=\nu D\). For the radial tangent \(\partial_\nu\Sigma=D\), the
gain \(\mathfrak G_\nu=D^{-1}/(\nu^2-\hbar^2/4)\) solves
Eq.~\eqref{eq:sldgain} because \(\Om D^{-1}\Om=-D\). Therefore
\(F_{\nu\nu}=\frac12\Tr(\mathfrak G_\nu D)=1/(\nu^2-\hbar^2/4)\).

The squeezing and rotation tangents are \(\partial_r\Sigma=2\nu\operatorname{diag}(e^{2r},-e^{-2r})\) and \(\partial_\alpha\Sigma=2\nu\sinh(2r)X\), where \(X\) is the off-diagonal Pauli matrix. The diagonal-traceless and off-diagonal ans\"atze for their gains yield \(F_{rr}=4\nu^2/(\nu^2+\hbar^2/4)\) and \(F_{\alpha\alpha}=4\nu^2\sinh^2(2r)/(\nu^2+\hbar^2/4)\). The trace inner products between these sectors vanish pairwise, meaning the QFI matrix is diagonal. Since \(g_{\rm B}=F_{\rm SLD}/4\), the radial line element is \(d\nu/[2\sqrt{\nu^2-\hbar^2/4}]\). Integrating this from \(\nu_0>\hbar/2\) to the boundary gives the finite radial Bures distance \(\ell=\frac12\operatorname{arcosh}(2\nu_0/\hbar)\).

\section{Frame dependence of the Schur sign subspaces}
\label{app:schur}

To see explicitly why the inertia calculation does not fix the sign
subspaces, consider the one-mode covariance
\begin{equation}
 \Sigma=\nu\operatorname{diag}(e^{2r},e^{-2r}),\qquad
 x=\frac{2\nu}{\hbar}>1,
\end{equation}
and use the ordered basis \((E_0,E_\times,E_+)\) from
Eq.~\eqref{eq:auxiliary-basis}.  With
\(c=\cosh(4r)\) and \(s=\sinh(4r)\),
\begin{equation}
 \mathcal F_\Sigma^*=8\nu^2
 \begin{pmatrix}c&0&s\\0&1&0\\s&0&c\end{pmatrix},
 \qquad K_f=\operatorname{diag}(1,-1,-1).
\end{equation}
Direct inversion of the second Schur complement gives
\begin{equation}
 G_{\rm vert}\bigl(\iota(E_+),\iota(E_+)\bigr)
 =\frac{x^2[\cosh(4r)-x^2]}{x^4-1}.
 \label{eq:squeezed-counterexample}
\end{equation}
Although \(\iota(E_+)\) lies in the fixed noncompact Cartan sector, this
quantity is positive whenever \(\cosh(4r)>x^2\).  The fixed basis vector can
therefore change sign while the total vertical inertia remains \((1,2)\).

\section*{Acknowledgments}
The author acknowledges the use of AI assistants for structural brainstorming, language refinement, and \LaTeX{} editing during preparation of the manuscript. The author bears full responsibility for all arguments, equations, and citations.

\bibliographystyle{unsrtnat}
\bibliography{foundation_references}

@book{Amari2016,
  author    = {Shun-ichi Amari},
  title     = {Information Geometry and Its Applications},
  publisher = {Springer},
  address   = {Tokyo},
  year      = {2016},
  doi       = {10.1007/978-4-431-55978-8}
}

@article{Banchi2015,
  author  = {Leonardo Banchi and Samuel L. Braunstein and Stefano Pirandola},
  title   = {Quantum Fidelity for Arbitrary Gaussian States},
  journal = {Physical Review Letters},
  volume  = {115},
  pages   = {260501},
  year    = {2015},
  doi     = {10.1103/PhysRevLett.115.260501}
}

@book{Bengtsson2017,
  author    = {Ingemar Bengtsson and Karol {\.{Z}}yczkowski},
  title     = {Geometry of Quantum States},
  edition   = {2},
  publisher = {Cambridge University Press},
  year      = {2017},
  doi       = {10.1017/9781139207010}
}

@article{Bhatia2019,
  author  = {Rajendra Bhatia and Tanvi Jain and Yongdo Lim},
  title   = {On the Bures--Wasserstein Distance Between Positive Definite Matrices},
  journal = {Expositiones Mathematicae},
  volume  = {37},
  number  = {2},
  pages   = {165--191},
  year    = {2019},
  doi     = {10.1016/j.exmath.2018.01.002}
}

@article{Braunstein1994,
  author  = {Samuel L. Braunstein and Carlton M. Caves},
  title   = {Statistical Distance and the Geometry of Quantum States},
  journal = {Physical Review Letters},
  volume  = {72},
  pages   = {3439--3443},
  year    = {1994},
  doi     = {10.1103/PhysRevLett.72.3439}
}

@article{Bures1969,
  author  = {Donald Bures},
  title   = {An Extension of Kakutani's Theorem on Infinite Product Measures to the Tensor Product of Semifinite $W^*$-Algebras},
  journal = {Transactions of the American Mathematical Society},
  volume  = {135},
  pages   = {199--212},
  year    = {1969},
  doi     = {10.1090/S0002-9947-1969-0236719-2}
}

@article{Chatterjee2026,
  author  = {Kaustav Chatterjee and Tanmoy Pandit and Varinder Singh and Pritam Chattopadhyay and Ulrik Lund Andersen},
  title   = {Even--Odd Splitting of the Gaussian Quantum Fisher Information: From Symplectic Geometry to Metrology},
  journal = {Quantum Science and Technology},
  volume  = {11},
  number  = {4},
  pages   = {045005},
  year    = {2026},
  doi     = {10.1088/2058-9565/ae92ac}
}

@book{deGosson2006,
  author    = {Maurice A. de Gosson},
  title     = {Symplectic Geometry and Quantum Mechanics},
  publisher = {Birkh\"auser},
  year      = {2006},
  doi       = {10.1007/3-7643-7575-2}
}

@article{Haynsworth1968,
  author  = {Emilie V. Haynsworth},
  title   = {Determination of the Inertia of a Partitioned Hermitian Matrix},
  journal = {Linear Algebra and Its Applications},
  volume  = {1},
  number  = {1},
  pages   = {73--81},
  year    = {1968},
  doi     = {10.1016/0024-3795(68)90050-5}
}

@book{Helgason1978,
  author    = {Sigurdur Helgason},
  title     = {Differential Geometry, Lie Groups, and Symmetric Spaces},
  publisher = {Academic Press},
  address   = {New York},
  year      = {1978}
}

@article{HuangWilde2026,
  author  = {Zixin Huang and Mark M. Wilde},
  title   = {Information Geometry of Bosonic Gaussian Thermal States},
  journal = {Physical Review A},
  volume  = {114},
  pages   = {022437},
  year    = {2026},
  doi     = {10.1103/t48x-g7xn}
}

@article{Petz1996,
  author  = {D\'enes Petz},
  title   = {Monotone Metrics on Matrix Spaces},
  journal = {Linear Algebra and Its Applications},
  volume  = {244},
  pages   = {81--96},
  year    = {1996},
  doi     = {10.1016/0024-3795(94)00211-8}
}

@article{PetzSudar1996,
  author  = {D\'enes Petz and Csaba Sud\'ar},
  title   = {Geometries of Quantum States},
  journal = {Journal of Mathematical Physics},
  volume  = {37},
  pages   = {2662--2673},
  year    = {1996},
  doi     = {10.1063/1.531535}
}

@article{Pinel2012,
  author  = {Olivier Pinel and Julien Fade and Daniel Braun and Pu Jian and Nicolas Treps and Claude Fabre},
  title   = {Ultimate Sensitivity of Precision Measurements with Intense Gaussian Quantum Light: A Multimodal Approach},
  journal = {Physical Review A},
  volume  = {85},
  pages   = {010101(R)},
  year    = {2012},
  doi     = {10.1103/PhysRevA.85.010101}
}

@article{Safranek2017,
  author  = {Dominik \v{S}afr\'anek},
  title   = {Discontinuities of the Quantum Fisher Information and the Bures Metric},
  journal = {Physical Review A},
  volume  = {95},
  pages   = {052320},
  year    = {2017},
  doi     = {10.1103/PhysRevA.95.052320}
}

@article{Safranek2019,
  author  = {Dominik \v{S}afr\'anek},
  title   = {Estimation of Gaussian Quantum States},
  journal = {Journal of Physics A: Mathematical and Theoretical},
  volume  = {52},
  pages   = {035304},
  year    = {2019},
  doi     = {10.1088/1751-8121/aaf068}
}

@article{Uhlmann1976,
  author  = {Armin Uhlmann},
  title   = {The Transition Probability in the State Space of a $^*$-Algebra},
  journal = {Reports on Mathematical Physics},
  volume  = {9},
  number  = {2},
  pages   = {273--279},
  year    = {1976},
  doi     = {10.1016/0034-4877(76)90060-4}
}

@article{Weedbrook2012,
  author  = {Christian Weedbrook and Stefano Pirandola and Ra\'ul Garc\'ia-Patr\'on and Nicolas J. Cerf and Timothy C. Ralph and Jeffrey H. Shapiro and Seth Lloyd},
  title   = {Gaussian Quantum Information},
  journal = {Reviews of Modern Physics},
  volume  = {84},
  pages   = {621--669},
  year    = {2012},
  doi     = {10.1103/RevModPhys.84.621}
}

@article{Williamson1936,
  author  = {John Williamson},
  title   = {On the Algebraic Problem Concerning the Normal Forms of Linear Dynamical Systems},
  journal = {American Journal of Mathematics},
  volume  = {58},
  number  = {1},
  pages   = {141--163},
  year    = {1936},
  doi     = {10.2307/2371062}
}

@book{Zhang2005,
  editor    = {Fuzhen Zhang},
  title     = {The Schur Complement and Its Applications},
  publisher = {Springer},
  address   = {New York},
  year      = {2005},
  doi       = {10.1007/b105056}
}

@article{Monras2013,
  author = {Monras, Alex},
  title = {Phase Space Formalism for Quantum Estimation of Gaussian States},
  journal = {arXiv preprint arXiv:1303.3682},
  year = {2013},
  eprint = {1303.3682},
  archivePrefix = {arXiv},
  primaryClass = {quant-ph}
}

@article{Malago2018,
  author = {Malag{\`o}, Luigi and Montrucchio, Luigi and Pistone, Giovanni},
  title = {Wasserstein Riemannian Geometry of Gaussian Densities},
  journal = {Information Geometry},
  volume = {1},
  number = {2},
  pages = {137--179},
  year = {2018},
  doi = {10.1007/s41884-018-0014-4}
}

\end{document}